\documentclass[a4paper,twocolumn,11pt,accepted=2025-10-10]{quantumarticle}
\pdfoutput=1
\usepackage{amsmath}
\usepackage{amssymb}
\usepackage{thm-restate}
\usepackage{thmtools}
\usepackage{breakurl}
\usepackage{graphicx}
\usepackage{dcolumn}
\usepackage{bm}
\usepackage{hyperref}

\usepackage{amsthm}
\usepackage{yquant}
\useyquantlanguage{groups}
\usepackage{tikz}
\tikzset{
  dot/.style={
    circle, fill=black, inner sep=1pt, outer sep=0pt
  }
}
\usetikzlibrary{fit}

\usepackage[ruled,linesnumbered,noline]{algorithm2e}
\usepackage{diagbox}
\usepackage{pgfplots}    
\pgfplotsset{width=12cm,compat=1.13}
\usepackage{cases}
\newtheorem{definition}{Definition}

\newtheorem{lemma}{Lemma}
\newtheorem{property}{Property}
\usepackage{float}
\usepackage{subfigure}
\usepackage{placeins}
\usepackage[numbers,sort&compress]{natbib}

\begin{document}


\title{Quantum circuit synthesis with SQiSW}

\author{Jialiang Tang}
 \email{tangjialiang20@mails.ucas.ac.cn}
\author{Jialin Zhang}%
 \email{zhangjialin@ict.ac.cn}
\author{Xiaoming Sun}%
 \email{sunxiaoming@ict.ac.cn}
\affiliation{%
 State Key Lab of Processors, Institute of Computing Technology, Chinese Academy of Sciences, Beijing 100190, China 
}%
\affiliation{%
 School of Computer Science and Technology, University of Chinese Academy of Sciences, Beijing 100049, China
}%

\begin{abstract}
The primary objective of quantum circuit synthesis is to efficiently and accurately realize specific quantum algorithms or operations utilizing a predefined set of quantum gates, while also optimizing the circuit size. It holds a pivotal position in Noisy Intermediate-Scale Quantum (NISQ) computation. Historically, most synthesis efforts have predominantly utilized CNOT or CZ gates as the 2-qubit gates. However, the SQiSW gate, also known as the square root of iSWAP gate, has garnered considerable attention due to its outstanding experimental performance with low error rates and high efficiency in 2-qubit gate synthesis. In this paper, we investigate the potential of the SQiSW gate in various synthesis problems by utilizing only the SQiSW gate along with arbitrary single-qubit gates, while optimizing the overall circuit size. For exact synthesis, the upper bound of SQiSW gates to synthesize arbitrary 3-qubit and $n$-qubit gates are 24 and $\frac{139}{192}4^n(1+o(1))$ respectively, which relies on the properties of SQiSW gate in Lie theory and Quantum Shannon Decomposition. We also introduce an exact synthesis scheme for Toffoli gate using only 8 SQiSW gates, which is grounded in numerical observation. More generally, with respect to numerical approximations, we provide a theoretical analysis of a pruning algorithm to reduce the size of the searching space in numerical experiment to $\frac{1}{12}+o(1)$ of previous size, helping us reach the result that 11 SQiSW gates are enough in arbitrary 3-qubit gates synthesis up to an acceptable numerical error.

\end{abstract}
\maketitle

\section{Introduction}\label{intro}
Quantum circuit synthesis is crucial to the implement of quantum algorithm on physical devices. It involves constructing a quantum circuit to realize the target unitary operator, while optimizing the circuit size or depth with respect to a given gate set \cite{Shende2006, Mottonen2004, jiang2022}. 

A lot of research has been done on quantum circuit synthesis, most of which are about exact synthesis and focus on CNOT (Controlled-NOT) gate at the early stage \cite{ Barenco1995,aho2003,Amy2018,yang2025,Cybenko2001,vatan2004,Knill1995,Mottonen2004,Shende2006,zi2025, guo}. The first synthesis algorithm using CNOT and arbitrary single-qubit gates was designed by Barenco et al. in 1995 \cite{Barenco1995}, giving an upper bound of $O(n^34^n)$ to synthesize arbitrary $n$-qubit gates. The upper bound was improved to $O(n4^n)$ in the same year \cite{Knill1995}. During the next decade, circuit transform techniques and Gray code were adopted in circuit synthesis to reduce the upper bound \cite{aho2003,Knill19952,Vartiainen2004}. Actually, the lower bound of arbitrary $n$-qubit gate was proved to be $\lfloor\frac{1}{4}(4^n-3n-1)\rfloor$ by parameter counting \cite{shende2004}. And the state-of-the-art result of upper bound is $\frac{23}{48}4^n-\frac{3}{2}2^n+\frac{4}{3}$ by making use of a brand new method called Quantum Shannon Decomposition \cite{Shende2006}. This upper bound is already no more than 2 times of the lower bound. On Toffoli gate, the CNOT cost is 6, which is a tight result \cite{shende2008, QCQI}. There are also several works to optimize CNOT-cost of multi-qubit Toffoli gate \cite{nie2024}. Recently, Chen et al. \cite{chen2023} looked into the universal synthesis using arbitrary 2-qubit gates. They proposed a ``rule-them-all'' exact synthesis scheme AshN, pointing out that the upper bound of circuit size of 2-qubit gates to synthesize $n$-qubit gate is $\frac{23}{64}4^n(1+o(1))$.

There are also also several studies on how to \textit{numerically approximate} the target unitary matrix, that is, how to optimize the circuit size required to approximate the target circuits within an acceptable numerical error~\cite{Goubault2019, Martinez2016, Cerezo2022, shirakawa2021, Ashhab2022, ashhab2023}. A study in 2019 \cite{Goubault2019} explored how to get a optimal synthesis of circuits numerically, which uses M{\o}lmer S{\o}rensen (MS) gate. Since MS gate entangles all the qubits in the circuit, the whole structure of circuit is fixed and determining the continuous parameters of different gates is the only thing. In 2022, Ashhab et al. \cite{Ashhab2022} designed a numerical optimization algorithm using CNOT gate and arbitrary single-qubit gates, achieving good performance in 2-qubit quantum state preparation, 3-qubit quantum state preparation and synthesis of $n$-qubit Toffoli gate. The number of CNOT gates to reach high fidelity is even close to the theoretical lower bound.

In addition, some works have proposed more general circuit synthesis and optimization frameworks that can support arbitrary set of gates. Anouk et al. raised Synthetiq \cite{synthetiq}, a fast and versatile quantum circuit synthesis framework based on Simulated Annealing. Addressing the limitations of existing tools, Synthetiq synthesizes circuits over arbitrary finite gate sets and supports partial specifications. Also, researchers from Berkeley developed a universal circuit compilation framework called BQSKit \cite{BQSKit, BQSKit_code}, which is a comprehensive and powerful circuit compilation tool. It is dedicated to providing efficient and scalable quantum circuit optimization and synthesis solutions. Its core goal is to reduce the depth of quantum circuits through innovative algorithms, improve operational reliability and computational efficiency, and is suitable for fields such as subroutine compilation, gate group conversion, and algorithm exploration.

However, other quantum gates besides CNOT gate are rarely studied in quantum circuit synthesis. A problem has then emerged that whether there exists a kind of quantum gate that has a more powerful synthesis ability and also shows off lower experimental error at the same time.

The SQiSW gate, known as the square root of iSWAP gate, is a 2-qubit quantum gate experimentally realized in earlier works \cite{Norbert2003,Bialczak2010,Abrams2020}. Recently it has been proved to potentially realize the dream \cite{Huang2023}. On one hand, SQiSW gate has a shorter gate time and lower error than CNOT gate on superconducting quantum processor of Alibaba. This new type of quantum gate set has excellent performance in experiments \cite{Huang2023}: The gate fidelity measured on single SQiSW gate is up to 99.72\%, with an average of 99.31\%. For any 2-qubit gate synthesis problem, it can achieve an average fidelity of 96.38\% on the random samples. Compared to using iSWAP gates on the same processor, the former reduces the average error by 41\%, while the latter reduces it by 50\%. On the other hand, SQiSW gate has an improved ability in synthesis of arbitrary 2-qubit gates than CNOT gate. Specifically, the upper bound of SQiSW gates and CNOT gates to synthesize arbitrary 2-qubit gates are both 3, and about 79\% 2-qubit gates can be synthesized by at most 2 SQiSW gates while the gates generated by 2 CNOT gates constitute a zero-measure set (see supplemental material of \cite{Huang2023}). These results rely on KAK decomposition and Weyl chamber, which are Lie algebra mathematical tools to depict general properties of $SU(4)$.  

Due to the superior properties of SQiSW gate over common CNOT gate, it's significant to look deeper into the potential of SQiSW gate in circuit synthesis, which hopefully help implement more efficient quantum circuits on superconducting quantum computers.

In this paper, we focus on the exact synthesis and numerical optimization with SQiSW circuits. We utilize only the
SQiSW gate along with arbitrary single-qubit gates, aiming at the optimization of number of SQiSW gates. Our main results are presented as follows.

\begin{restatable}{theorem}{thmone}\label{thm1}
    An arbitrary 3-qubit gate can be synthesized using a maximum of 24 SQiSW gates.
\end{restatable}

\begin{restatable}{theorem}{thmtwo}\label{thm2}
    An arbitrary $n$-qubit gate can be synthesized using a maximum of $\frac{139}{192}\times 4^n - 3\times 2^n + \frac{5}{3}$ SQiSW gates.
\end{restatable}

Our approach to deriving Theorem \ref{thm1} is that we firstly decompose arbitrary 3-qubit gates into arbitrary 2-qubit gates,  and subsequently utilizing the properties of SQiSW gates in synthesis of arbitrary 2-qubit gates. To prove Theorem \ref{thm2}, we employ Quantum Shannon Decomposition recursively in which the base case is Theorem \ref{thm1}, with several circuit optimization techniques.

Algorithm \ref{alg:numerical toffoli} and algorithm \ref{alg:numerical 3qubit} are numerical optimization algorithms for Toffoli gate and arbitrary 3-qubit gates respectively. Algorithm \ref{alg3} is a pruning algorithm using two pruning techniques, which are qubit re-arrangement and circuit reversion \cite{Ashhab2022, synthetiq}, to speed up our numerical optimization algorithms.

The numerical results indicate that we only need 8 and 11 SQiSW gates in circuit to synthesize Toffoli gate and arbitrary 3-qubit gates respectively under acceptable numerical error, which is intuitively shown in Fig.~\ref{fig3} and Fig.~\ref{fig4}. Theorem \ref{thm3} is about the efficiency analysis of algorithm \ref{alg:numerical toffoli} and algorithm \ref{alg:numerical 3qubit} with application of algorithm \ref{alg3} on them.

\begin{restatable}{theorem}{thmthree}\label{thm3}
    The pruning algorithm \ref{alg3} reduces the size of circuit structure space (with 3 qubits and N quantum gates) in algorithm \ref{alg:numerical toffoli} and algorithm \ref{alg:numerical 3qubit} to $\frac{1}{12}\times3^{N}+3^{\lfloor\frac{N-1}{2}\rfloor}+\frac{1}{4}=\frac{1}{12}3^{N}(1+o(1))$.

\end{restatable}

\begin{restatable}{theorem}{thmfour}\label{thm4}
    Toffoli gate can be exactly synthesized by 8 SQiSW gates, a scheme shown in Fig.~\ref{fig5}.
\end{restatable}

Theorem \ref{thm4} is proved in a numerical-aided way. It is the numerical optimization on Toffoli gate that provides the observation of pattern of parameters in circuit. By repeatedly guessing, fixing and retraining, we get to have a circuit with only few parameters, which approximates Toffoli gate well. Then we prove it can be an exact synthesis scheme by assigning appropriate values to the parameters.

The rest of this paper is organized as follows. Section \ref{sec2} is about preliminaries. In section \ref{sec3}, we prove Theorem \ref{thm1} and recursively induce Theorem \ref{thm2}, our road map contains some decomposition schemes and circuit optimization techniques. In section \ref{sec4}, we show the numerical optimization algorithm and pruning algorithm. Then we prove Theorem \ref{thm3} and show the numerical results. In section \ref{sec5}, we prove Theorem \ref{thm4}, which is derived by numerical observation. Section \ref{sum} makes a summary.


\section{Preliminaries}\label{sec2}
Given a matrix $U$, we denote its $(i,j)$-th entry as $U_{ij}$. We also denote the conjugate transpose
of $U$ as $U^\dagger$. And we use $\otimes$ to denote the Kronecker tensor product over two matrices.

We denote the unitary group $U_n(\mathbb{C})$ as $U(n)$ and specialized unitary group $SU_n(\mathbb{C})$ as $SU(n)$. Here $\mathbb{C}$ is the field of complex numbers.
\subsection{Quantum multiplexor}
\begin{definition}[Quantum multiplexor \cite{Shende2006}]
    A quantum multiplexor with $s$ controlled qubits (positioned at the highest levels) and $d$ target qubits is a block-diagonal unitary matrix comprising $2^s$ unitary matrices $U_i\in U(2^d)$, $1\leq i \leq2^s$:
    \begin{equation*}
    U = 
    \begin{pmatrix}
        U_1&& \\
        &\ddots&\\
        &&U_{2^s}
    \end{pmatrix}.
\end{equation*}
\end{definition}

Here is a typical example: the multiplexor-$R_y$, which is a multiplexor with 1 data qubit acted by $R_y$ gate with different angles depending to the state of the controlled qubits. In the circuit diagram, we mark each controlled qubit of $U$ in quantum circuits with a ``$\boxed{}$'' symbol.
\begin{center}
    \scalebox{1.5}{
    \begin{tikzpicture}
    \begin{yquantgroup}
    \registers{
    qubit {} q[2];
    
    }
    \circuit{
    slash q[0];
    hspace {2mm} q[0];
    [this control/.append style={shape=yquant-rectangle,draw, /yquant/default fill}]
    box {$R_y$} q[1]| q[0];
    hspace {2mm} q[0];
    }
    \end{yquantgroup}
    \end{tikzpicture}
    }
\end{center}

\begin{lemma}[Decomposition of multiplexor rotation gates \cite{Shende2006,bullock2003,Mottonen2004}]For multiplexor-$R_k$ and $k=y,z$, it holds that

    \begin{center}
    \scalebox{1.5}
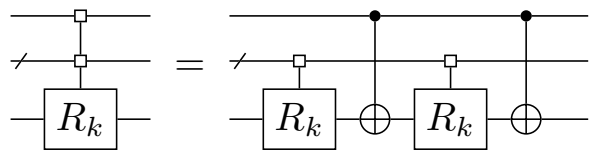
{
    \begin{tikzpicture}
    \begin{yquantgroup}
    \registers{
    qubit {} q[3];
    
    }
    \circuit{
    slash q[1];
    hspace {2mm} q[1];
    [this control/.append style={shape=yquant-rectangle,draw, /yquant/default fill}]
    box {$R_k$} q[2]| q[0],q[1];
    hspace {2mm} q[1];
    }
    \equals
    \circuit{
    slash q[1];
    hspace {2mm} q[1];
    [this control/.append style={shape=yquant-rectangle,draw, /yquant/default fill}]
    box {$R_k$} q[2]|q[1];
    cnot q[2] | q[0];
    [this control/.append style={shape=yquant-rectangle,draw, /yquant/default fill}]
    box {$R_k$} q[2]|q[1];
    cnot q[2] | q[0];
    hspace {2mm} q[1];
    }
    \end{yquantgroup}
    \end{tikzpicture}
    }
    \end{center}
    and any $n$-qubit multiplexor-$R_k$ can be synthesized by at most $2^{n-1}$ CNOT gates. 
    \label{lemma1}
\end{lemma}

\subsection{Local equivalence and Weyl chamber}
Since single-qubit gates contribute little to the cost of a quantum circuit, it is essential to study the equivalence between quantum gates that only differ in single-qubit gates on sides. Fortunately, KAK decomposition and Weyl chamber provide mathematical tools to depict 2-qubit gates up to single-qubit gates \cite{tucci2005,byron,cartan,Huang2023,wyle,geo}. 
\begin{lemma}[KAK decomposition \cite{tucci2005}]
    $\forall U\in SU(4)$, there exists a unique $\Vec{k}=(x,y,z)$, $\frac{\pi}{4}\geq x\geq y\geq |z|$, single-qubit gates $A_0, A_1, B_0, B_1\in SU(2)$ and $g\in \{1, i\}$ s.t. 
    \begin{equation*}
        U = g\cdot (A_0\otimes A_1) exp\{i\Vec{k}\cdot \Vec{\sigma}\} (B_0\otimes B_1),
    \end{equation*}
    in which $\Vec{\sigma} \equiv [X\otimes X, Y\otimes Y, Z\otimes Z]$. 
\end{lemma}
The vector $\Vec{k}$ is called \emph{interaction coefficients}, which characterize the equivalent class of an arbitrary $U\in SU(4)$, denoted as $k(U)$. The 3-dimension area related to $\Vec{k}$ is known as \emph{Weyl chamber} \cite{tucci2005,wyle,Huang2023}. The Weyl chamber $W$ is defined as
$
    W \equiv \{\frac{\pi}{4}\geq x\geq y\geq z~~\mathrm{and}~~z\geq 0~~\mathrm{if}~~x = \frac{\pi}{4}~|~(x,y,z)\in \mathbb{R}^3\},
$

shown in Fig.~\ref{fig1} with some common gates and their interaction coefficients marked.

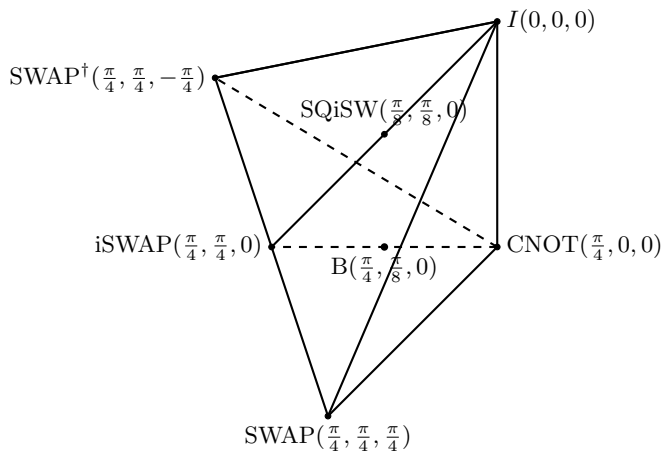
\begin{figure*}[!htbp]
    \centering
    \begin{tikzpicture}[scale = 1.1]
    \draw[thick] (2,4)node[right]{$I(0,0,0)$}--(2,0)node[right]{$\mathrm{CNOT(\frac{\pi}{4},0,0)}$}--(-1,-3)node[below]{$\mathrm{SWAP(\frac{\pi}{4},\frac{\pi}{4},\frac{\pi}{4})}$}--(-2,0)node[left]{$\mathrm{iSWAP(\frac{\pi}{4},\frac{\pi}{4},0)}$}--(-3,3)node[left]{$\mathrm{SWAP^\dagger(\frac{\pi}{4},\frac{\pi}{4},-\frac{\pi}{4})}$}--(2,4);
    \draw[thick] (2,4)--(-1,-3);
    \draw[thick] (2,4)--(-2,0);
    \draw[thick] (2,4)--(-3,3);
    \draw[thick,dashed] (2,0)--(-2,0);
    \draw[thick,dashed] (2,0)--(-3,3);
    \draw[thick] (0,0)node[below]{$\mathrm{B(\frac{\pi}{4},\frac{\pi}{8},0)}$};
    \draw[thick] (0,2)node[above]{$\mathrm{SQiSW(\frac{\pi}{8},\frac{\pi}{8},0)}$};
    \draw[fill] (-2,0) circle (.05);
    \draw[fill] (2,0) circle (.05);
    \draw[fill] (2,4) circle (.05);
    \draw[fill] (-3,3) circle (.05);
    \draw[fill] (-1,-3) circle (.05);
    \draw[fill] (0,2) circle (.05);
    \node[dot] (0,0) {};
    
    \end{tikzpicture}
    
    \caption{Weyl chamber in $\mathbb{R}^3$, with some common gates and their interaction coefficients. Note that there are many ways to draw Weyl chamber, and here we adopt that in \cite{Huang2023}.}
    \label{fig1}

\end{figure*}

\begin{definition}[Local equivalence \cite{Huang2023}]
    Two unitary matrices $U$ and $V$ are said to be locally equivalent, if and only if $U$ and $V$ share the same interaction coefficients in Weyl chamber.
\end{definition}
The following properties are of utmost importance when addressing interaction coefficients.

\begin{property}[Transformation of interaction coefficients \cite{tucci2005}]\label{prop1}
    The local equivalent class of a unitary matrix $U$ remains the same under such transformation on its interaction coefficients:
    \begin{itemize}
        \item Add or minus one element of $(x,y,z)$ by $\pi/2$. For example, add $x$ by $\pi/2$ while keep $y$ and $z$:
        \begin{equation*}
            (x,y,z) \mapsto (x+\pi/2,y,z).
        \end{equation*}
        \item Take the opposite number of two elements of $(x,y,z)$. For example, take the opposite number of $x$ and $y$, while keep $z$:
        \begin{equation*}
            (x,y,z) \mapsto (-x,-y,z).
        \end{equation*}
        \item Exchange two elements of $(x,y,z)$. For example, exchange $x$ and $y$, while keep $z$:
        \begin{equation*}
            (x,y,z) \mapsto (y,x,z).
        \end{equation*}
    \end{itemize}
\end{property}
In fact, if we take the equivalence under transformation as $\approx$, then \cite{tucci2005,geo}
\begin{equation*}
    W = \mathbb{R}^3 / \approx.
\end{equation*}

So we can easily round the interaction coefficients of any arbitrary $U\in SU(4)$ to one and only one position in Weyl chamber. 

\subsection{SQiSW gate and its properties}
SQiSW gate is defined as the matrix square-root of iSWAP gate, that is,
\begin{align*}
\mathrm{SQiSW}
&\equiv
\mathrm{\sqrt{iSWAP}}\\
&=
\begin{pmatrix}
1&0&0&0 \\
0&\frac{1}{\sqrt{2}}&\frac{i}{\sqrt{2}}&0\\
0&\frac{i}{\sqrt{2}}&\frac{1}{\sqrt{2}}&0 \\
0&0&0&1
\end{pmatrix}.
\end{align*}
Besides, SQiSW has such mathematical properties which are useful in the rest part of this paper. Here we skip the proof since it's easy to verify them.
\begin{property}[Mathematical Properties of SQiSW]\label{prop2}
SQiSW gate has such properties:
    \begin{itemize}
    \item SQiSW gate is commutative with $Z\otimes Z$, thus commutative with $R_z(\theta)\otimes R_z(\theta)$.
    \item SQiSW gate is qubit symmetrical. It satisfies that SQiSW = SWAP $\cdot$ SQiSW $\cdot$ SWAP.
    \item SQiSW$^\dagger$ is locally equivalent to SQiSW.
    \end{itemize}
\end{property}

It's been proved that any arbitrary $U\in SU(4)$ can be synthesized by at most 3 SQiSW gates \cite{Huang2023}, and the region that can be spanned by 2 SQiSW gates is stained red in Fig.~\ref{fig2}. The rest part needs at least 3 SQiSW gates. We rewrite it as the Lemma \ref{lemma3} below.

\begin{figure*}[!htbp]
    \centering
    \begin{tikzpicture}[scale = 1.1]
    \filldraw[red] (2,4)--(-2,0)--(0.5,-1.5)--(2,0)--(2,4);
    \draw[thick] (2,4)node[right]{$I(0,0,0)$}--(2,0)node[right]{$\mathrm{CNOT(\frac{\pi}{4},0,0)}$}--(-1,-3)node[below]{$\mathrm{SWAP(\frac{\pi}{4},\frac{\pi}{4},\frac{\pi}{4})}$}--(-2,0)node[left]{$\mathrm{iSWAP(\frac{\pi}{4},\frac{\pi}{4},0)}$}--(-3,3)node[left]{$\mathrm{SWAP^\dagger(\frac{\pi}{4},\frac{\pi}{4},-\frac{\pi}{4})}$}--(2,4);
    \draw[thick] (2,4)--(-1,-3);
    \draw[thick] (2,4)--(-2,0);
    \draw[thick] (2,4)--(-3,3);
    \draw[thick,dashed] (2,0)--(-2,0);
    \draw[thick,dashed] (2,0)--(-3,3);
    \draw[thick,dashed] (2,4)--(-0.4,1.44)node[above left]{$(\frac{\pi}{4},\frac{\pi}{8},-\frac{\pi}{8})$};
    \draw[thick,dashed] (-2,0)--(-0.4,1.44);
    \draw[thick,dashed] (2,4)--(0.5,-1.5)node[right]{$(\frac{\pi}{4},\frac{\pi}{8},\frac{\pi}{8})$};
    \draw[thick,dashed] (-2,0)--(0.5,-1.5);
    \draw[thick] (0,0)node[below]{$\mathrm{B(\frac{\pi}{4},\frac{\pi}{8},0)}$};
    \draw[fill] (-0.4,1.44) circle (.05);
    \draw[fill] (0.5,-1.5) circle (.05);
    \draw[fill] (-2,0) circle (.05);
    \draw[fill] (2,0) circle (.05);
    \draw[fill] (2,4) circle (.05);
    \draw[fill] (-3,3) circle (.05);
    \draw[fill] (-1,-3) circle (.05);
    \node[dot] (0,0) {};
    
    \end{tikzpicture}
    \caption{Region in Weyl chamber that can be spanned by 2 SQiSW gates, filled in red \cite{Huang2023}.}
    \label{fig2}
\end{figure*}
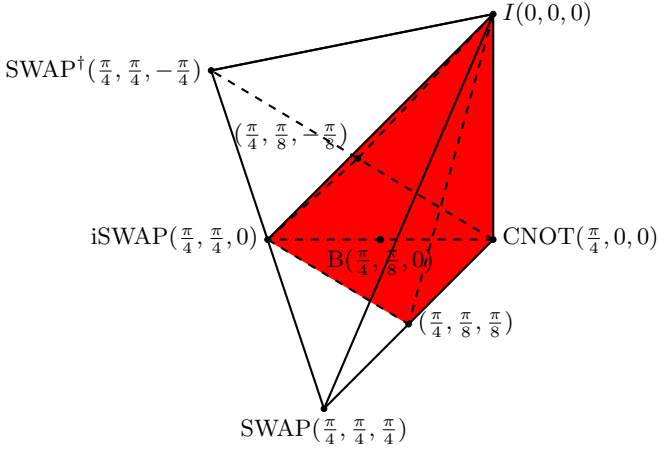

\begin{lemma}[Synthesis of 2-qubit gates using SQiSW \cite{Huang2023}]
    Any arbitrary 2-qubit gate $U\in SU(4)$ can be exactly synthesized by at most 3 SQiSW gates, up to single qubit gates. $U$ can be synthesized by 1 SQiSW gate, if and only if $k(U)$ is $(\frac{\pi}{8},\frac{\pi}{8},0)$. $U$ can be synthesized by at most 2 SQiSW gates, if and only if $k(U)$ is in $W^{\prime}$, where $
        W^{\prime} \equiv \{\frac{\pi}{4}\geq x\geq y\geq |z|~\mathrm{and}~x\geq y+|z|~|~(x,y,z)\in W\}.
    $

    \label{lemma3}
\end{lemma}
This indicates that 3 SQiSW gates can span the whole Weyl chamber, and 2 SQiSW gates span the red area of Weyl chamber in Fig.~\ref{fig2}.

\section{Exact synthesis of arbitrary quantum gates}\label{sec3}
In this section, we prove that the upper bound of arbitrary 3-qubit synthesis is 24, then recursively derive the upper bound of arbitrary $n$-qubit synthesis, which is $\frac{139}{192}4^n(1+o(1))$. 
\subsection{Exact synthesis of arbitrary 3-qubit gates}

We first introduce two lemmas, namely Lemma~\ref{lemma5} and Lemma~\ref{lemma6}, as the foundation for the proof of Theorem~\ref{thm1}.

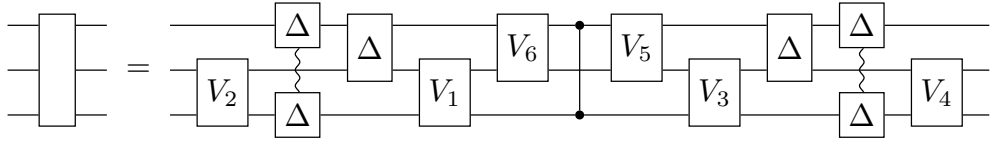
\begin{lemma}[Decomposition of arbitrary 3-qubit gates \cite{chen2023}]\label{lemma5}
    Any 3-qubit gate can be synthesized by at most 11 2-qubit gates, a scheme shown as below.
    \begin{widetext}
    \begin{center}
    \scalebox{1.2}{
    \begin{tikzpicture}
        \begin{yquantgroup}
            \registers{
            qubit {} q[3];
            }
            \circuit{
            hspace {2mm} q[1];
            box {} (q[0-2]);
            hspace {2mm} q[1];
            }
            \equals
            \circuit{
            hspace {2mm} q[0-2];
            box {$V_2$} (q[1-2]);
            
            hspace {2mm} q[0-2];
            box {$\Delta$} (q[0,2]);
            
            hspace {2mm} q[0-2];
            box {$\Delta$} (q[0-1]);
            hspace {2mm} q[0-2];
            box {$V_1$} (q[1-2]);
            hspace {2mm} q[0-2];
            box {$V_6$} (q[0-1]);

            zz (q[0,2]);
            box {$V_5$} (q[0-1]);
            hspace {2mm} q[0-2];
            box {$V_3$} (q[1-2]);
            hspace {2mm} q[0-2];
            box {$\Delta$} (q[0-1]);
            hspace {2mm} q[0-2];
            box {$\Delta$} (q[0,2]);
            hspace {2mm} q[0-2];
            box {$V_4$} (q[1-2]);
            hspace {2mm} q[0-2];
            }
        \end{yquantgroup}
    \end{tikzpicture}
    }
\end{center}
\end{widetext}
\end{lemma}

\begin{lemma}\label{lemma6}
Any arbitrary 2-qubit gate can be synthesized using 2 SQiSW gates and 1 diagonal 2-qubit gate. The synthesis scheme is shown below, where the SQiSW gate is denoted as S.
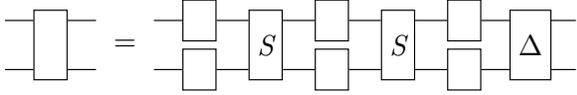
    \begin{center}
    \scalebox{1.1}{
    \begin{tikzpicture}
    \begin{yquantgroup}
    \registers{
    qubit {} q[2];
    }
    \circuit{
    hspace {2mm} q[1];
    [this control/.append style={shape=yquant-rectangle,draw, /yquant/default fill}]
    box {} (q[0],q[1]);
    hspace {2mm} q[1];
    }
    \equals
    
    \circuit{
    hspace {2mm} q[0-1];
    box {} (q[0]);
    box {} (q[1]);
    hspace {2mm} q[0-1];
    box {S} (q[0-1]);
    hspace {2mm} q[0-1];
    box {} (q[0]);
    box {} (q[1]);
    hspace {2mm} q[0-1];
    box {S} (q[0-1]);
    hspace {2mm} q[0-1];
    box {} (q[0]);
    box {} (q[1]);
    hspace {2mm} q[0-1];
    box {$\Delta$} (q[0,1]);
    hspace {2mm} q[0-1];
    }

    \end{yquantgroup}
    \end{tikzpicture}
    }
\end{center}
\end{lemma}

\begin{proof}
    Firstly, the arbitrary 2-qubit gates can be synthesized by 2 CNOT gates and 1 diagonal 2-qubit gate, by \cite{Shende2006}. So we have
    \begin{center}
    \scalebox{1.1}
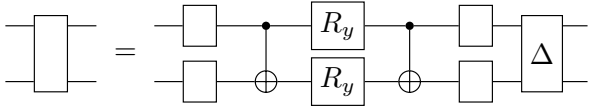
{
    \begin{tikzpicture}
    \begin{yquantgroup}
    \registers{
    qubit {} q[2];
    }
    \circuit{
    hspace {2mm} q[1];
    [this control/.append style={shape=yquant-rectangle,draw, /yquant/default fill}]
    box {} (q[0],q[1]);
    hspace {2mm} q[1];
    }
    \equals
    
    \circuit{
    hspace {2mm} q[0-1];
    box {} (q[0]);
    box {} (q[1]);
    hspace {2mm} q[0-1];
    cnot q[1] | q[0];
    hspace {2mm} q[0-1];
    box {$R_y$} (q[0]);
    box {$R_y$} (q[1]);
    hspace {2mm} q[0-1];
    cnot q[1] | q[0];
    hspace {2mm} q[0-1];
    box {} (q[0]);
    box {} (q[1]);
    hspace {2mm} q[0-1];
    box {$\Delta$} (q[0,1]);
    hspace {2mm} q[0-1];
    }

    \end{yquantgroup}
    \end{tikzpicture}
    }
\end{center}

Inside such scheme, at least 2 elements of interaction coefficients of the sub-circuit
    \begin{center}
    \scalebox{1.2}
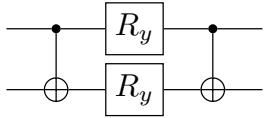
{
    \begin{tikzpicture}
    \begin{yquantgroup}
    \registers{
    qubit {} q[2];
    }
    \circuit{

    hspace {2mm} q[0-1];
    cnot q[1] | q[0];
    hspace {2mm} q[0-1];
    box {$R_y$} (q[0]);
    box {$R_y$} (q[1]);
    hspace {2mm} q[0-1];
    cnot q[1] | q[0];
    hspace {2mm} q[0-1];
    }
    \end{yquantgroup}
    \end{tikzpicture}
    }
\end{center}
are zeros (by the general calculation process for interaction coefficients in section 3.2 of \cite{byron}). So the sub-circuit above can be synthesized by 2 SQiSW gates.
\end{proof}

Now we have all the fragment contributing to Theorem \ref{thm1}. We now give the complete  Theorem \ref{thm1}.

\thmone*
\begin{proof}
     First, we aim to prove that all diagonal gates in the scheme of Lemma \ref{lemma5} can be synthesized by at most 2 SQiSW gates. 
    
    Any diagonal matrix is locally equivalent to a $R_{zz}$ gate, which is defined as
    \begin{equation*}
        R_{zz}(\theta) = R_z(\theta) \otimes R_z(-\theta),
    \end{equation*}
    so we have
    \begin{equation*}
        R_{zz}(\theta) = 
        \begin{pmatrix}
            e^{-i \frac{\theta}{2}}&&&\\
            &e^{i \frac{\theta}{2}}&&\\
            &&e^{i \frac{\theta}{2}}&\\
            &&&e^{-i \frac{\theta}{2}}
        \end{pmatrix}.
    \end{equation*}
    
    It's easy to see
    \begin{equation*}
        R_{zz}(\theta) = exp(-i \frac{\theta}{2} Z\otimes Z),
    \end{equation*}
    which tells us that the interaction coefficients of $R_{zz}$ gate is $(0,0,-\frac{\theta}{2})$. By Lemma \ref{lemma3}, $R_{zz}$ can be synthesized by 2 SQiSW gates. Thus all arbitrary diagonal gates can be synthesized by 2 SQiSW gates. 
    
    Next, we prove $V_5$ and $V_6$ can be synthesized by at most 2 SQiSW gates. In the original proof of Lemma \ref{lemma5} (see Theorem 12 in \cite{chen2023}), we get to know that $V_5$ and $V_6$ have the form like:
    \begin{center}
    \scalebox{1.2}
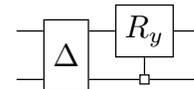
{
    \begin{tikzpicture}
        \begin{yquant}
            qubit {} q[2];
            hspace {2mm} q[0-1];
            box {$\Delta$} (q[0-1]);
            hspace {2mm} q[0-1];
            [this control/.append style={shape=yquant-rectangle,draw, /yquant/default fill}]
            box {$R_y$} q[0] | q[1];
            hspace {2mm} q[0-1];
        \end{yquant}
    \end{tikzpicture}
    }
    \end{center}
    
    Rearrange qubits and reverse the circuit, we only need to prove the circuit
    \begin{center}
    \scalebox{1.2}
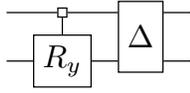
{
    \begin{tikzpicture}
        \begin{yquant}
            qubit {} q[2];
            hspace {2mm} q[0-1];
            [this control/.append style={shape=yquant-rectangle,draw, /yquant/default fill}]
            box {$R_y$} q[1] | q[0];
            hspace {2mm} q[0-1];
            box {$\Delta$} (q[0-1]); 
            hspace {2mm} q[0-1];
        \end{yquant}
    \end{tikzpicture}
    }
    \end{center}
    can be synthesized by at most 2 SQiSW gates. That's due to the second and third points of Property \ref{prop2}.
    
    Since an arbitrary diagonal gate is locally equivalent to $R_{zz}$ gate up to single qubit gates on the leftmost side or rightmost side, we only need to prove the interaction coefficients of unitary
    \begin{equation*}
        \begin{pmatrix}
            R_z(\theta)&\\
            &R_z(-\theta)
        \end{pmatrix}
        \begin{pmatrix}
            R_y(\theta_1)&\\
            &R_y(\theta_2)
        \end{pmatrix}
        \label{eq:3-25}
    \end{equation*}
    locates in $W^\prime$, which is spanned by 2 SQiSW gates.

    To simplify the calculation, we consider the unitary locally equivalent to it. By applying $R_z$ and $R_y$ gates to two sides of it, we have
    \begin{equation*}
        \begin{split}
        & 
        \begin{pmatrix}
            R_z(\theta)&\\
            &R_z(-\theta)
        \end{pmatrix}
        \begin{pmatrix}
            R_y(\theta_1)&\\
            &R_y(\theta_2)
        \end{pmatrix}\\
        &\sim 
        \begin{pmatrix}
            I&\\
            &R_z(-2\theta)
        \end{pmatrix}
        \begin{pmatrix}
            R_y(\theta_1)&\\
            &R_y(\theta_2)
        \end{pmatrix}\\
        &\sim
        \begin{pmatrix}
            I&\\
            &R_z(-2\theta)
        \end{pmatrix}
        \begin{pmatrix}
            I&\\
            &R_y(\theta_2-\theta_1)
        \end{pmatrix}\\
        &=
        \begin{pmatrix}
            I&\\
            &R_z(-2\theta)R_y(\theta_2-\theta_1)
        \end{pmatrix},
        \end{split}
    \end{equation*}

    where $\sim$ means being locally equivalent to.

    For simplicity's sake, we assume
    \begin{equation*}
        U = 
        \begin{pmatrix}
            I&\\
            &R_z(\theta_1)R_y(\theta_2)
        \end{pmatrix}.
    \end{equation*}
    
    The interaction coefficients of $U$ has 2 elements being zero. So $U$ lies on the x-axis in Weyl chamber, thus $U$ can be synthesized by at most 2 SQiSW gates. So $V_5$ and $V_6$ can all be synthesized by at most 2 SQiSW gates.

    Then, we make use of Lemma \ref{lemma6} to optimize the circuit. Synthesizing $V_1$ and $V_3$ by the scheme in Lemma \ref{lemma6} allows us to move the diagonal gates from $V_1$ and $V_3$ across other diagonal gates in Lemma \ref{lemma5}. Then they can be finally absorbed into $V_4$ and $V_6$. 
    
    Also, We can synthesize the middle CNOT gate in circuit of Lemma \ref{lemma5} by 2 SQiSW gates, which is shown in Fig.~\ref{fig2}. To conclude, we can synthesize any arbitrary 3-qubit gate by at most 24 SQiSW gates.
\end{proof}

Recall that the upper bound of synthesis of arbitrary 3-qubit gates using CNOT gate is 20 \cite{Shende2006}. Our result is 24 for SQiSW gate, 40\% better than trivially replacing each CNOT gate by 2 SQiSW gates.
\subsection{Exact synthesis of arbitrary $n$-qubit gates}
We already have the upper bound of synthesis of arbitrary 3-qubit gates, now we take Theorem \ref{thm1} as the base case and prove Theorem \ref{thm2}, using Quantum Shannon Decomposition recursively and some circuit optimization techniques.

\begin{lemma}[Cosine-Sine Decomposition \cite{Mottonen2004,Shende2006,tucci1998}]
\label{lemma7}
Any arbitrary $n$-qubit gate can be exactly synthesized as below.
\begin{center}
    \scalebox{1.5}{
    \begin{tikzpicture}
    \begin{yquantgroup}
    \registers{
    qubit {} q[2];
    
    }
    \circuit{
    slash q[1];
    hspace {2mm} q[1];
    box {} (q[0],q[1]);
    hspace {2mm} q[1];
    }
    \equals
    \circuit{
    slash q[1];
    hspace {2mm} q[1];
    [this control/.append style={shape=yquant-rectangle,draw, /yquant/default fill}]
    box {} q[1]|q[0];
    [this control/.append style={shape=yquant-rectangle,draw, /yquant/default fill}]
    box {$R_y$} q[0]|q[1];
    [this control/.append style={shape=yquant-rectangle,draw, /yquant/default fill}]
    box {} q[1]|q[0];
    hspace {2mm} q[1];
    }
    \end{yquantgroup}
    \end{tikzpicture}
    }
\end{center}
where $n\geq 2$.
\end{lemma}

\begin{lemma}[Demultiplexing of multiplexor \cite{Shende2006}]
\label{lemma8}
    Any arbitrary $n$-qubit multiplexor with 1 controlled qubit can be exactly synthesized as below.
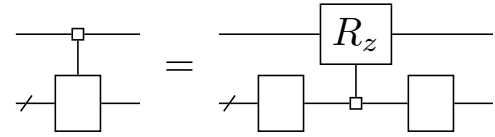
\begin{center}
    \scalebox{1.5}{
    \begin{tikzpicture}
    \begin{yquantgroup}
    \registers{
    qubit {} q[2];
    }
    \circuit{
    slash q[1];
    hspace {2mm} q[1];
    [this control/.append style={shape=yquant-rectangle,draw, /yquant/default fill}]
    box {} q[1]|q[0];
    hspace {2mm} q[1];
    }
    \equals
    \circuit{
    slash q[1];
    hspace {2mm} q[1];
    [this control/.append style={shape=yquant-rectangle,draw, /yquant/default fill}]
    box {} q[1];
    [this control/.append style={shape=yquant-rectangle,draw, /yquant/default fill}]
    box {$R_z$} q[0]|q[1];
    [this control/.append style={shape=yquant-rectangle,draw, /yquant/default fill}]
    box {} q[1];
    hspace {2mm} q[1];
    }
    \end{yquantgroup}
    \end{tikzpicture}
    }
\end{center}
where $n\geq 2$.
\end{lemma}

We can get a synthesis scheme from the combination of Lemma \ref{lemma7} and \ref{lemma8}, which is called Quantum Shannon Decomposition.

\begin{lemma}[Quantum Shannon Decomposition \cite{Shende2006}]
\label{lemma9}
Any arbitrary $n$-qubit gate can be exactly synthesized as below.
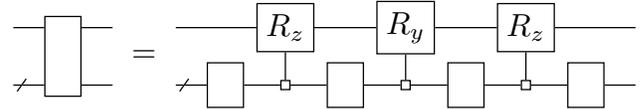
\begin{center}
    \scalebox{1.15}{
    \begin{tikzpicture}
    \begin{yquantgroup}
    \registers{
    qubit {} q[2];
    }
    \circuit{

    slash q[1];
    hspace {2mm} q[1];
    [this control/.append style={shape=yquant-rectangle,draw, /yquant/default fill}]
    box {} (q[1],q[0]);
    hspace {2mm} q[1];
    }
    \equals
    \circuit{

    slash q[1];
    hspace {2mm} q[1];
    [this control/.append style={shape=yquant-rectangle,draw, /yquant/default fill}]
    box {} q[1];
    [this control/.append style={shape=yquant-rectangle,draw, /yquant/default fill}]
    box {$R_z$} q[0]|q[1];
    [this control/.append style={shape=yquant-rectangle,draw, /yquant/default fill}]
    box {} q[1];
    [this control/.append style={shape=yquant-rectangle,draw, /yquant/default fill}]
    box {$R_y$} q[0]|q[1];
    [this control/.append style={shape=yquant-rectangle,draw, /yquant/default fill}]
    box {} q[1];
    [this control/.append style={shape=yquant-rectangle,draw, /yquant/default fill}]
    box {$R_z$} q[0]|q[1];
    [this control/.append style={shape=yquant-rectangle,draw, /yquant/default fill}]
    box {} q[1];
    hspace {2mm} q[1];
    }
    \end{yquantgroup}
    \end{tikzpicture}
    }
\end{center}
where $n\geq 2$.
\end{lemma}
Quantum Shannon Decomposition gives a recursive decomposition scheme for arbitrary $n$-qubit gates, thus we have enough evidence to prove Theorem \ref{thm2}.

Denote $c_l$ as the upper bound of number of SQiSW gates to synthesize any arbitrary $l$-qubit gate. Now we state the proof of Theorem \ref{thm2} below.
\thmtwo*
\begin{proof}
    According to Lemma \ref{lemma9} and Lemma \ref{lemma1}, we have
    \begin{equation*}
        c_j \leq 4c_{j-1} + 3\times 2^j,~j\geq 2,
    \end{equation*}
    by replacing each CNOT gate by 2 SQiSW gates.

    Then 
    \begin{equation*}
        c_n \leq 4^{n-l}(c_l+3\times 2^l) - 3\times 2^n.
    \end{equation*}

    Now we introduce two circuit optimization techniques to decrease the upper bound.

    First, notice that CZ gate is locally equivalent to CNOT gate (see Fig.~\ref{fig:cnot}), and Lemma \ref{lemma1} still holds after replacing all CNOT gates by CZ gates.

    So the multiplexor-$R_y$ gates in Lemma \ref{lemma7} can be synthesized by CZ gates with one CZ gate located on the leftmost side or rightmost side. Since CZ gate is diagonal, it can be absorbed into the adjacent multiplexor, saving $(4^{n-l}-1)/3$ CNOT (or CZ) gates in total, thereby saving $2(4^{n-l}-1)/3$ SQiSW gates.

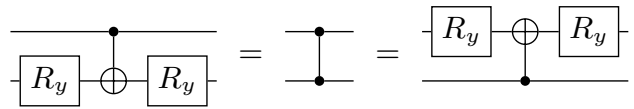
\begin{figure}[htbp]
    \begin{center}
    \scalebox{1.2}{
    \begin{tikzpicture}
    \begin{yquantgroup}
    \registers{
    qubit {} q[2];
    }
    \circuit{
    box {$R_y$} q[1];
    cnot q[1] | q[0];
    box {$R_y$} q[1];
    }
    \equals
    
    \circuit{
    zz (q[0],q[1]);
    }
    \equals
    
    \circuit{
    box {$R_y$} q[0];
    cnot q[0] | q[1];
    box {$R_y$} q[0];
    }
    \end{yquantgroup}
    \end{tikzpicture}
    }
    \end{center}
    \caption{CNOT gate is locally equivalent to CZ gate, and a CZ gate is equal to a CNOT gate with $R_y$ gates on both sides of its target qubit.}
    \label{fig:cnot}
\end{figure}
    
    Second, stop recursion at $l=3$ and consider Theorem \ref{thm1}. We can synthesize $V_2$ or $V_4$ by Lemma \ref{lemma6} and get to have a synthesis scheme for 3-qubit gates with a diagonal gate on the leftmost or rightmost side. All these diagonal gates are commutative with multiplexor gates through controlled qubits in Lemma \ref{lemma9}. Then they can be absorbed into the neighboring 3-qubit gate. We get to save $4^{n-3}-1$ SQiSW gates in total.

    Since we already have $c_3 = 24$, such two optimization techniques above help achieve an upper bound that
    \begin{equation*}
    \begin{split}
        c_n &\leq 4^{n-3}(c_3+3\times 2^3) - 3\times 2^n - 5(4^{n-3}-1)/3\\
        &= \frac{139}{192}\times 4^n - 3\times 2^n + \frac{5}{3}.
    \end{split}
    \end{equation*}

    Here, we have finished the proof of Theorem \ref{thm2}.
\end{proof}

Recall that the upper bound of synthesis of arbitrary $n$-qubit gates using CNOT gate is $\frac{23}{48}4^n(1+o(1))$ \cite{Shende2006}. Our result is $\frac{139}{192}4^n(1+o(1))$ using SQiSW gate, 24\% better than trivially replacing each CNOT gate by 2 SQiSW gates when $n$ is large.
\section{Numerical optimization and its analysis}\label{sec4}
Different from exact synthesis, numerical optimization aims to find a circuit numerically approximating the target circuit in synthesis problems. In this section, we run numerical optimization to synthesize arbitrary 3-qubit gates and Toffoli gate, studying how many SQiSW gates are sufficient to reach low enough error in the two tasks. Additionally, we give a reliable pruning algorithm and analyze its performance which is shown as Theorem \ref{thm3}.

In this section we denote $n$ as the number of qubits, $N$ as the number of SQiSW gates in the circuit. For the target function of optimization, we employ the known standard: the red{``distance''} or error between two circuits $U$ and $V$ is defined as
\begin{equation*}
    E(U,V) = 1-\frac{|tr(U^\dagger V)|}{2^n}.
\end{equation*}

The stopping threshold of error in the optimization algorithm is set to $10^{-6}$, which is nothing compared to the experimental error of realizing a quantum circuit.

To better describe the algorithm, we first introduce the concept of circuit structure.
\begin{definition}[Circuit structure \cite{Ashhab2022}]
    The circuit structure of a quantum circuit is a sequence of positions to show the configuration of multi-qubit gates in the circuit.
\end{definition}
In this paper, we use $0\sim n-1$ to show the position in a circuit with $n$ qubits and in chronological order. For example, structure of the circuit in Fig.~\ref{fig:struc} is $(0,1), (1,2), (0,2)$.

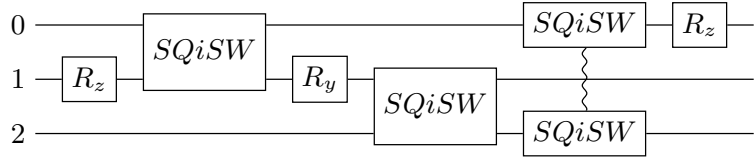
\begin{figure*}[!htbp]
    \begin{center}
    \scalebox{1.2}{
    \begin{tikzpicture}
        \begin{yquant}
            qubit {$\idx$} q[3];

            hspace {2mm} q[0-2];
            box {$R_z$} (q[1]);
            hspace {2mm} q[0-2];
            box {$SQiSW$} (q[0-1]);
            hspace {2mm} q[0-2];
            box {$R_y$} (q[1]);
            hspace {2mm} q[0-2];
            box {$SQiSW$} (q[1-2]);
            hspace {2mm} q[0-2];
            box {$SQiSW$} (q[0],q[2]);
            hspace {2mm} q[0-2];
            box {$R_z$} (q[0]);
            hspace {2mm} q[0-2];
        \end{yquant}
    \end{tikzpicture}
    }
    \end{center}

    \caption{A quantum circuit with structure $(0,1), (1,2), (0,2)$.}
    \label{fig:struc}

\end{figure*}

\subsection{Numerical optimization algorithm}
We should look for the minimum number of SQiSW gates used in the circuit to numerically approximate target circuits. But we may have lots of structures to consider, and there may be a huge gap between performances of different structures. 

We adopt the frame of two searching spaces, which are mentioned in \cite{Ashhab2022}. More specifically, given the number of SQiSW gates, for each circuit structure (circuit structure space), we optimize the parameters in the circuit (parameter space) to approximate target circuit. When programming, they are actually two nested ``for'' loops. And we utilize qfactor package in python \cite{qfactor} to help the parameter optimization, which is a parameter learning tool. In the rest part of this paper, we take it as a black box. The algorithms for Toffoli gate and arbitrary 3-qubit gates are shown in algorithm \ref{alg:numerical toffoli} and algorithm \ref{alg:numerical 3qubit}. Note that SQiSW gate is a deterministic gate, we only need to optimize parameters of single-qubit gates. In algorithm \ref{alg:numerical 3qubit}, we sample unitary matrices uniformly under Haar measure \cite{Haar} and study the average error.

\begin{algorithm}[htbp]
\caption{Numerical optimization for Toffoli gate.} 
\label{alg:numerical toffoli}
\AlgoDontDisplayBlockMarkers\SetAlgoNoEnd\SetAlgoNoLine
\SetNlSty{textbf}{}{\enspace}
\SetKwFor{For}{for}{\string:}{}
\SetKwIF{If}{ElseIf}{Else}{if}{:}{elif}{else:}{}%
\KwOut{Deterministic circuit $C$ numerically approximating Toffoli gate.}
\BlankLine
\For{num of SQiSW gates $N$}{
\quad\For{circuit structure $S$}{
\quad\quad\For{$i$ in range(repetition\_time)}{
\quad\quad\quad randomize($\Vec{\theta}$)\;
\quad\quad\quad $\Vec{\theta}$ = Learn($S$, $\Vec{\theta}$, Toffoli, learning parameters)\;
\quad\quad\quad $E$ = dist($S(\Vec{\theta})$, Toffoli)\;
\quad\quad\quad \If{$E\leq\epsilon$}{
\quad\quad\quad\quad\textbf{Return} $S(\Vec{\theta})$\;
                }
        }
    }
}  
\end{algorithm}

\begin{algorithm}[htbp]
\caption{Numerical optimization for arbitrary 3-qubit gates.} 
\label{alg:numerical 3qubit}
\AlgoDontDisplayBlockMarkers\SetAlgoNoEnd\SetAlgoNoLine
\SetNlSty{textbf}{}{\enspace}
\SetKwFor{For}{for}{\string:}{}
\SetKwIF{If}{ElseIf}{Else}{if}{:}{elif}{else:}{}%
\KwOut{Circuit structure $S$ numerically approximating arbitrary 3-qubit gates under acceptable average error.}
\BlankLine
\For{num of SQiSW gates $N$}{
\quad\For{circuit structure $S$}{
\quad\quad lst\_U = sample($SU(8)$,100)\;
\quad\quad error = [~]\;
\quad\quad\For{$U$ in lst\_U}{
\quad\quad\quad min\_d = +$\infty$\;

\quad\quad\quad\For{$i$ in range(repetition\_time)}{
\quad\quad\quad\quad randomize($\Vec{\theta}$)\;
\quad\quad\quad\quad $\Vec{\theta}$ = Learn($S$, $\Vec{\theta}$, Toffoli, learning parameters)\;
\quad\quad\quad\quad\If{dist($S(\Vec{\theta}), U$) \textless min\_d}{
\quad\quad\quad\quad\quad min\_d = dist($S(\Vec{\theta}), U$)\;
                    }
                }
\quad\quad\quad error.append(min\_d)\;
                }
\quad\quad average = avg(error)\;
\quad\quad \If{average $\leq\epsilon$}{
\quad\quad\quad\textbf{Return} $S$\;
        }
    }
}  
\end{algorithm}

\subsection{Pruning techniques for speed-up}
In a $n$-qubit system, each SQiSW gate has $n(n-1)/2$ possible positions, which are $(0,1)$, $(0,2)$ and $(1,2)$ when $n=3$. Recall SQiSW gate is qubit symmetrical, the qubit order of SQiSW gate doesn't matter.

That means we would have a structure space of $3^N$ size, leading to huge difficulties when $N$ is large, which is intuitively shown in table \ref{tab:space}.

The idea is to find these circuit structures that are ``equivalent'' under numerical optimization, thus we only need to keep one element of each class with little influence on the performance of our numerical optimization algorithms.

\begin{table}[h]
    \centering 
    \caption{The size of circuit structure space under different number of qubits $n$ and number of SQiSW gates $N$}
    \label{tab:space}
    \renewcommand{\arraystretch}{1} 
    \begin{tabular}{|c|c|c|c|c|} 
        \hline
        \diagbox{n}{N} & $5$ & $10$ & $15$ & $50$\\
        \hline
        $2$ & $1$ & $1$  & $1$  & $1$ \\
        \hline
        $3$ & $243$ & $59049$ & $14348907$ & $> 10^{23}$  \\
        \hline
        $4$ & $7776$ & $60466176$ & $\sim 10^{11}$ &$> 10^{36}$\\
        \hline
        $5$ & $10^5$ & $10^{10}$ & $10^{15}$ & $10^{50}$  \\
        \hline
    \end{tabular}
\end{table}

We adopt two pruning techniques which are usually used in numerical optimization: qubit rearrangement and circuit reversion \cite{Ashhab2022,synthetiq}, and give an explicit theoretical analysis of pruning efficiency when both of them are used in our synthesis tasks. To make our pruning algorithm theoretically reliable, we are supposed to make sure that such techniques make little influence on the numerical results. Here we define numerical equivalence and equivalence closure to help our proof of Theorem \ref{thm3}.

\begin{definition}[Numerical equivalence]
    Given circuit structures $C_1$ and $C_2$, if 
    \begin{equation*}
        \forall T \in \mathcal{T},~\exists~\Vec{\theta}~s.t.~(C_1, \Vec{\theta})\xrightarrow{} T
    \end{equation*}
    is equivalent to
    \begin{equation*}
        \forall T \in \mathcal{T},~\exists~\Vec{\theta^\prime}~s.t.~(C_2, \Vec{\theta^\prime})\xrightarrow{} T, 
    \end{equation*}
    where $\mathcal{T}$ is our target gate set, 
    then we say $C_1$ and $C_2$ are numerically equivalent (or $C_1$ and $C_2$ are in the same equivalent class under numerical optimization). 
\end{definition}
Our target gate set in algorithm \ref{alg:numerical toffoli} and algorithm \ref{alg:numerical 3qubit} are $\{\mathrm{Toffoli}\}$ and $SU(8)$ respectively. 
To better depict the equivalent class under more than one equivalence relation, we define equivalent closure.

\begin{definition}[Equivalent closure]
    An equivalent closure $\mathcal{C}$ under operation set $R$ is a maximal set s.t. $\forall~e_1, e_2 \in \mathcal{C}$, $e_1$ can be generated by applying operations in $R$ on $e_2$ finite times. Here $R$ must guarantee that $e_1$ and $e_2$ are numerically equivalent. 
\end{definition}

It can be verified that in algorithm \ref{alg:numerical toffoli} and algorithm \ref{alg:numerical 3qubit}, the circuit structures before and after being transformed by $R$ are numerically equivalent. Here $R$ contains ``qubit rearrangement'' and ``circuit reversion'', in which qubit rearrangement means rearranging the qubits order of the circuit structure and circuit reversion means reversing the gate order of the circuit. 

So the whole circuit structure space in both algorithm \ref{alg:numerical toffoli} and algorithm \ref{alg:numerical 3qubit} can be divided into different equivalent closures under $R$. Now we state algorithm \ref{alg3}, which successfully keeps exact one structure remained but kicks other structures out in each equivalent closure under $R$ when it is applied to algorithm \ref{alg:numerical toffoli} and algorithm \ref{alg:numerical 3qubit}.

\begin{algorithm}[htbp]
\caption{Pruning of circuit structure space \cite{Ashhab2022,synthetiq}} 
\label{alg3}
\AlgoDontDisplayBlockMarkers\SetAlgoNoEnd\SetAlgoNoLine
\SetNlSty{textbf}{}{\enspace}
\SetKwFor{For}{for}{\string:}{}
\SetKwIF{If}{ElseIf}{Else}{if}{:}{elif}{else:}{}%
\KwIn{Set of all possible circuit structures $S$.}
\KwOut{Set $S$ with circuit structures remained after qubit rearrangement and circuit reversion.}
\BlankLine
\For{circuit structure $C\in S$}{
\quad $S_1$ $\leftarrow$ circuit structures equivalent to $C$ under qubit rearrangement\;

\quad $S_2$ $\leftarrow$ circuit structures equivalent to $C$ under circuit reversion\;

\quad Remove circuits different from $C$ in $S_1 \cup S_2$ from $S$\;
}  
\end{algorithm}


Previously, there was no theoretical analysis conducted. Below, we will conduct a asymptotical analysis using both pruning methods simultaneously in numerical optimization.

To prove Theorem \ref{thm3}, we need to give explicit analysis of these closures, stated as below. Denote that $f_{inv}$ is a function mapping a circuit structure space to its reversed structure, and $f_{arr}$ is a function mapping a circuit structure to the set of circuit structures that only differ in qubit arrangement.
\thmthree*
\begin{proof}
    Firstly we prove that there are only equivalent closures with 3 or 6 or 12 elements under $R$.

    For circuit structure with gates having only one position, its closure only has 3 elements.
    For circuit structure with at least 2 positions of gates, it has 5 other different elements under qubit rearrangement. 
    If a structure $C^{\prime}$ generated by reversing a circuit in rearrangement of original circuit $C$ is equal to a structure in rearrangement of original circuit $C$, then the rearrangement of $C^{\prime}$ is exact the arrangement of $C$.
    
    So either the reversion of arrangement of $C$ is the same as arrangement of $C$ or it has no same element with arrangement of $C$. In the first case, we have a equivalent closure with 6 elements, and in the other case we have a equivalent closure with 12 elements. They are the maximal set since we cannot generate a new element by apply any equivalent relation on any element of it.

    The rest part of proof is on the calculation of number of each kind of equivalent closure. We only have one 3-element closure. The question is, how many circuits $C$ are there satisfying 
    \begin{equation*}
        f_{inv}(C) \in f_{arr}(C),
    \end{equation*}
    where $C$ has at least two different positions (that is, $C$ is not an all-$(0,1)$ or all-$(0,2)$ or all-$(0,3)$ pattern).
    
    The rearrangement function set is isomorphic to the permutation group
    \begin{equation*}
        S_6 = \{f_1, f_2, f_3, \tau_1, \tau_2, \tau_3\},
    \end{equation*}
    in which
    \begin{equation*}
        f_1 = \binom{0~1~2}{0~1~2},~
        f_2 = \binom{0~1~2}{2~0~1},~
        f_3 = \binom{0~1~2}{1~2~0}, 
    \end{equation*}
    and
    \begin{equation*}
        \tau_1 = \binom{0~1~2}{1~0~2},~
        \tau_2 = \binom{0~1~2}{2~1~0},~
        \tau_3 = \binom{0~1~2}{0~2~1}.
    \end{equation*}
    
    Without causing confusion, we take $S_6$ to represent the rearrangement function set, and each element in $S_6$ refers to the corresponding qubit rearrangement operation.
    
    Then we calculate number of circuit $C$ satisfying
    \begin{equation}\label{eq:prune}
        f_{inv}(C) = f(C),~f\in S_6.
    \end{equation}
    We go into two situations:
    
    \textit{Case 1.} Here $N$ is even. 
    Assume $N = 2n$ and $C = a_1a_2\dots a_{2n}$ (according to the reading order of a quantum circuit), each $a_i$ means a position of entangled gate and can be $(0,1)$, $(0,2)$ or $(1,2)$. 
    
    If $f= f_1$, then $inv(C) = C$, there are $3^n-3$ circuits, since each $a_i = a_{2n+1-i}$ has 3 options and we exclude 3 circuits in the 3-element closure. 
    
    If $f = f_2$, we have
    \begin{equation*}
    \begin{aligned}
        a_{2n+1-i} &= f(a_i), \\
        a_i &= f(a_{2n+1-i}),
    \end{aligned}
    \end{equation*}
    for $1\leq i \leq 2n$.
    That is,
    \begin{equation*}
    \begin{aligned}
        a_i &= f(f(a_i)),
    \end{aligned}
    \end{equation*}
    for $1\leq i \leq 2n$.
    But it doesn't hold for any $a_i$. So there's no circuit satisfying (\ref{eq:prune}).
    
    If $f = f_3$, similarly there's no circuit satisfying (\ref{eq:prune}). 
    
    If $f = \tau_1$, then 
    \begin{equation*}
        \begin{aligned}
            (a_i, a_{2n+1-i}) &= ((0,1),(0,1)),\\
            or~(a_i, a_{2n+1-i}) &= ((0,2),(1,2)),\\
            or~(a_i, a_{2n+1-i}) &= ((1,2),(0,2)),
        \end{aligned}
    \end{equation*}
    for $1\leq i \leq n$.
    We have $3^n-1$ circuits, excluding the circuit with only $(0,1)$.
    
    If $f = \tau_2$, then 
    \begin{equation*}
        \begin{aligned}
            (a_i, a_{2n+1-i}) &= ((0,2),(0,2)),\\
            or~(a_i, a_{2n+1-i}) &= ((0,1),(1,2)),\\
            or~(a_i, a_{2n+1-i}) &= ((1,2),(0,1)),
        \end{aligned}
    \end{equation*}
    for $1\leq i \leq n$.
    We have $3^n-1$ circuits, excluding the circuit with only $(0,2)$.
    
    If $f = \tau_3$, then 
    \begin{equation*}
        \begin{aligned}
            (a_i, a_{2n+1-i}) &= ((1,2),(1,2)),\\
            or~(a_i, a_{2n+1-i}) &= ((0,1),(0,2)),\\
            or~(a_i, a_{2n+1-i}) &= ((0,2),(0,1)),
        \end{aligned}
    \end{equation*}
    for $1\leq i \leq n$.
    We have $3^n-1$ circuits, excluding the circuit with only $(1,2)$.

    So we have 1 closure with 3 elements, $\frac{1}{6}(4\times 3^n -6)$ closures with 6 elements, and $\frac{11}{12}[3^{2n}-(4\times 3^n -6) -3]$ closures with 12 elements. Algorithm \ref{alg3} successfully reduces the size of space to 
    \begin{equation*}
        \frac{1}{12}\times 3^{N} + 3^{\frac{N}{2}-1} + \frac{1}{4}.
    \end{equation*}

    \textit{Case 2.} Here $N$ is odd.
    Assume $N = 2n+1$ and $C = a_1a_2\dots a_{2n+1}$.
    
    If $f = f_1$, then
    \begin{equation*}
        \begin{aligned}
            (a_i, a_{2n+2-i}) &= ((0,1),(0,1)),\\
            or~(a_i, a_{2n+2-i}) &= ((0,2),(0,2)),\\
            or~(a_i, a_{2n+2-i}) &= ((1,2),(1,2)),
        \end{aligned}
    \end{equation*}
    for $1\leq i \leq n$. And $a_{n+1}$ can be arbitrarily chosen. Thus we have $3^{n+1}-3$ circuits. $f=f_2$ and $f=f_3$ are all similar to corresponding situations that $N$ is even. 
    
    If $f = \tau_1$, then 
    \begin{equation*}
        \begin{aligned}
            (a_i, a_{2n+1-i}) &= ((0,1),(0,1)),\\
            or~(a_i, a_{2n+1-i}) &= ((0,2),(1,2)),\\
            or~(a_i, a_{2n+1-i}) &= ((1,2),(0,2)),
        \end{aligned}
    \end{equation*}
    for $1\leq i \leq n$. And $a_{n+1} = (0,1)$.
    We have $3^n-1$ circuits, excluding the circuit with only $(0,1)$.
    
    If $f = \tau_2$, then 
    \begin{equation*}
        \begin{aligned}
            (a_i, a_{2n+1-i}) &= ((0,2),(0,2)),\\
            or~(a_i, a_{2n+1-i}) &= ((0,1),(1,2)),\\
            or~(a_i, a_{2n+1-i}) &= ((1,2),(0,1)),
        \end{aligned}
    \end{equation*}
    for $1\leq i \leq n$. And $a_{n+1} = (0,2)$.
    We have $3^n-1$ circuits, excluding the circuit with only $(0,2)$.
    
    If $f = \tau_3$, then 
    \begin{equation*}
        \begin{aligned}
            (a_i, a_{2n+1-i}) &= ((1,2),(1,2)),\\
            or~(a_i, a_{2n+1-i}) &= ((0,1),(0,2)),\\
            or~(a_i, a_{2n+1-i}) &= ((0,2),(0,1)),
        \end{aligned}
    \end{equation*}
    for $1\leq i \leq n$. And $a_{n+1} = (1,2)$.
    We have $3^n-1$ circuits, excluding the circuit with only $(1,2)$.

    So we have 1 closure with 3 elements, $\frac{1}{6}(2\times 3^{n+1} -6)$ closures with 6 elements, and $\frac{11}{12}[3^{2n}-(2\times 3^{n+1} -6) -3]$ closures with 12 elements. Algorithm \ref{alg3} successfully reduces the size of space to 
    \begin{equation*}
        \frac{1}{12}\times 3^{N} + \frac{1}{2}\times 3^{\lfloor\frac{N}{2}\rfloor} + \frac{1}{4}.
    \end{equation*}
    
    Combining the two cases, we have finished the proof of Theorem \ref{thm3}.
\end{proof}

Theorem \ref{thm3} tells that the pruning algorithm \ref{alg3} reduces the size of circuit structure space to $\frac{1}{12}+o(1)$ of previous size, significantly increasing the efficiency of algorithm \ref{alg:numerical toffoli} and algorithm \ref{alg:numerical 3qubit}.


\subsection{Our results and analysis}
We apply the pruning algorithm and set the repetition time to 10 in both Algorithm~\ref{alg:numerical toffoli} and Algorithm~\ref{alg:numerical 3qubit}. Then, we present the results of the numerical optimization in a line chart, shown in Figure~\ref{fig3} and Figure~\ref{fig4}. It is worth noting that the reason we start with $N=5$ in algorithm \ref{alg:numerical toffoli} and algorithm \ref{alg:numerical 3qubit} is that a Toffoli gate needs at least 5 2-qubit gates \cite{Yu_2013} and lower bound of 2-qubit gates synthesizing any arbitrary 3-qubit gate is 6 \cite{shende2004,Shen2004}.
\begin{figure}[htbp]
    \begin{center}
    \subfigure[best error to approximate Toffoli gate]{
    \label{fig3}
    \scalebox{0.6}{
    \begin{tikzpicture}                 
    \begin{axis}[
    xlabel={N},
    ylabel={log(E)}
    ]                                   
    \addplot+[sharp plot]               
    coordinates                         
    {                                   
     (5,-1.9) (6,-2.3) (7,-2.5)
     (8,-10.9) (9,-11.8) (10,-12.4)
    };
    \end{axis}                         
    \end{tikzpicture} 
    }
    }
    
    \subfigure[best average error to approximate arbitrary 3-qubit gates]{
    \label{fig4}
    \scalebox{0.6}{
    \begin{tikzpicture}                 
    \begin{axis}[
    xlabel={N},
    ylabel={log(E)}
    ]                                   
    \addplot+[sharp plot]               
    coordinates                         
    {                                   
     (5,-1.2) (6,-2.3) (7,-2.7)
     (8,-2.9) (9,-3.2) (10,-3.7)
     (11,-6.1) (12,-6.8) 
    };
    \end{axis}                         
    \end{tikzpicture}  
    }
    }
    \end{center}

    \caption{The logarithm of best error E in numerical optimization, using $N$ SQiSW gates in the circuit.}

\end{figure}
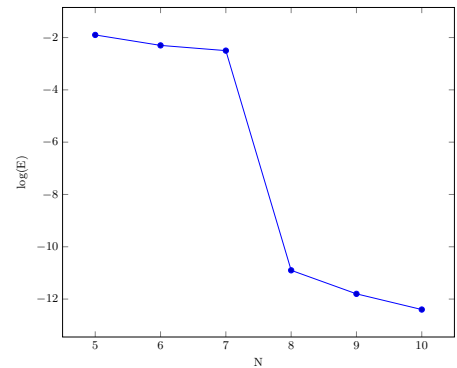
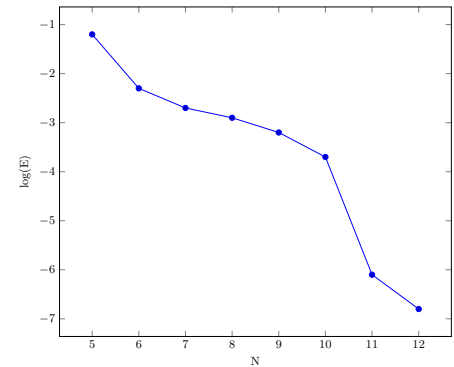

Fig.~\ref{fig3} and Fig.~\ref{fig4} indicate that 8 and 11 SQiSW gates are sufficient to numerically approximate Toffoli gate and arbitrary 3-qubit gates respectively.

As a comparison, we only need 11 SQiSW gates to numerically approximate arbitrary 3-qubit gates, and the number for CNOT gates is 14 \cite{chen2023}. Thus SQiSW gate works better with aspect to circuit size in numerical optimization of arbitrary 3-qubit gates. For Toffoli gate, though we need 2 more entangled gates compared to CNOT gate \cite{Ashhab2022}, a single SQiSW gate only contributes to half of the experimental error of a single CNOT gate \cite{Huang2023}. On the other hand, the numerical result can even help us derive an exactly synthesis scheme for Toffoli gate using 8 SQiSW gates, which is introduced in detail in section \ref{sec5}.

\section{Exact synthesis of Toffoli gate}\label{sec5}
In this section, we prove that 8 SQiSW gates are sufficient to exactly synthesize Toffoli gate and give a feasible synthesis scheme, which is based on the observation of parameters in numerical optimization.

The idea is repeatedly rounding and retraining the rest parameters. For the rounding, it's all human work to select the suitable value. We round the numerical parameters to a near constant number. As an example, we round 3.14158 to $\pi$, 1.57106 to $\pi/2$ and fix them and retrain the other parameters without obvious pattern at this stage. 

Sometimes it's not enough to just round the parameter to a constant number. We also try to observe and utilize the inner relation between different parameters. For instance, if parameters $x$ and $y$ satisfy $y=x+3.14267$ throughout the training, then we believe $y = x+\pi$. This is another rounding strategy.

After such process, we finally derive the circuit with only one free parameter, shown in Fig.~\ref{fig5}. Such a scheme can reach an error of even $10^{-16}$ in numerical optimization to approximate Toffoli gate, which is a strong evidence for the exact synthesis of Toffoli gate with 8 SQiSW gates. 

Now we prove Theorem \ref{thm4}, by assigning appropriate value to the parameter remained and prove the existence of solution of the circuit equations.

\begin{figure*}[!htbp]
\begin{center}
    \scalebox{0.55}{
    \begin{tikzpicture}
    \begin{yquant}
    qubit {} q[3];
    box {$R_y(\frac{\pi}{2})$} q[1];
    box {$R_z(\frac{\pi}{2})$} q[2];
    box {$R_z(-\frac{\pi}{2})$} q[1];
    box {$SQiSW$} (q[1-2]);
    box {$R_y(\frac{\pi}{2})$} q[2];
    box {$R_z(\frac{\pi}{2})$} q[2];
    box {$SQiSW$} (q[0],q[2]);
    box {$SQiSW$} (q[0],q[2]);
    box {$R_y(\frac{\pi}{2})$} q[0];
    box {$R_y(\theta_1)$} q[2];
    box {$SQiSW$} (q[0],q[1]);
    box {$R_y(\frac{\pi}{4})$} q[0];
    box {$R_y(-\frac{3\pi}{4})$} q[1];
    box {$R_z(\frac{\pi}{2})$} q[0];
    box {$R_z(\frac{\pi}{2})$} q[1];
    box {$SQiSW$} (q[1],q[2]);
    box {$R_z(-\pi)$} q[1];
    box {$R_y(\theta_2)$} q[2];
    box {$SQiSW$} (q[1],q[2]);
    box {$R_y(\frac{\pi}{2})$} q[1];
    box {$R_y(\theta_3)$} q[2];
    box {$R_z(-\pi)$} q[2];
    box {$SQiSW$} (q[0],q[2]);
    box {$SQiSW$} (q[0],q[2]);
    box {$R_y(\pi)$} q[0];
    box {$R_y(-\frac{\pi}{4})$} q[2];
    box {$R_z(-\frac{\pi}{4})$} q[0];
    box {$R_z(\frac{\pi}{2})$} q[2];
    \end{yquant}
    \end{tikzpicture}
    }
\end{center}
\caption{A synthesis scheme for Toffoli gate, using 8 SQiSW gates and several single-qubit gates, where $\theta_3=\theta_2+\pi/2=\theta_1+\pi$.}
\label{fig5}
\end{figure*}
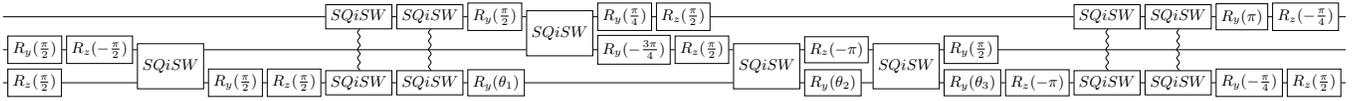

\thmfour*
\begin{proof}
    Replace $\theta_1$ by $x$. Calculate the matrix representation $U$ of the circuit in Fig.~\ref{fig5}, we have
\begin{widetext}
\begin{align*}
        U_{11} &= U_{22} = (-\frac{1}{4}+\frac{1}{4}i)(-1)^{\frac{3}{8}}(\cos{\frac{x}{2}}+\sin{\frac{x}{2}})(-2+i\sqrt{2}+(2+\sqrt{2})\sin{x}), \\
        U_{33} &= U_{44} = \frac{1}{2}(-1)^{\frac{7}{8}}(\cos{\frac{x}{2}}+\sin{\frac{x}{2}})(-i-\sqrt{2}+(1+\sqrt{2})\sin{x}),\\
        U_{55} &= U_{66} = -\frac{(-1)^{\frac{1}{8}}(\cos{\frac{x}{2}}+\sin{\frac{x}{2}})((5-i)-(4-i)\sqrt{2}+(-1-i+\sqrt{2})\sin{x})}{2(-1+(-1)^{\frac{1}{4}})^3},\\
        U_{87} &= U_{78} = -\frac{1}{2}(-1)^{\frac{1}{8}}(\cos{\frac{x}{2}}+\sin{\frac{x}{2}})(i-\sqrt{2}+(1+\sqrt{2})\sin{x}).
\end{align*} 
\end{widetext}

The rest elements of $U$ are either 0 or a multiple of $(\cos{\frac{x}{2}}-\sin{\frac{x}{2}})(1+(1+\sqrt{2})\sin{x})$. So we derive equations by making $U$ equal to the matrix representation of Toffoli gate:
\begin{widetext}
\begin{numcases}{}
(-\frac{1}{4}+\frac{1}{4}i)(-1)^{\frac{3}{8}}(\cos{\frac{x}{2}}+\sin{\frac{x}{2}})(-2+i\sqrt{2}+(2+\sqrt{2})\sin{x})    &= 1  \nonumber \\
\frac{1}{2}(-1)^{\frac{7}{8}}(\cos{\frac{x}{2}}+\sin{\frac{x}{2}})(-i-\sqrt{2}+(1+\sqrt{2})\sin{x})      &= 1  \nonumber \\
-\frac{(-1)^{\frac{1}{8}}(\cos{\frac{x}{2}}+\sin{\frac{x}{2}})((5-i)-(4-i)\sqrt{2}+(-1-i+\sqrt{2})\sin{x})}{2(-1+(-1)^{\frac{1}{4}})^3}   &= 1 \nonumber \\
-\frac{1}{2}(-1)^{\frac{1}{8}}(\cos{\frac{x}{2}}+\sin{\frac{x}{2}})(i-\sqrt{2}+(1+\sqrt{2})\sin{x})               &= 1 \nonumber \\
\cos{\frac{x}{2}} = \sin{\frac{x}{2}}\quad or\quad\sin{x} = 1-\sqrt{2}.\label{eq}
\end{numcases}
\end{widetext}

It can be verified that when the second equation of (\ref{eq}) holds, and the solution of it is in $(-\frac{\pi}{2},\pi)$, we have
\begin{numcases}{}
U_{11}*U_{87} &= 1 \nonumber  \\
U_{33}*U_{87} &= 1 \nonumber  \\
U_{33}*U_{55} &= 1 \nonumber  \\
U_{87} &= 1. \nonumber 
\end{numcases}

So the circuit equations have an analytical solution, and a feasible solution is $x = \arcsin{\left(1-\sqrt{2}\right)}$, where $-\pi/2 < x < 0$.

Thus, the circuit in Fig.~\ref{fig5} is exactly equal to Toffoli gate, when
\begin{equation*}
    \theta_3 = \theta_2 + \frac{\pi}{2} = \theta_1 + \pi = \arcsin{\left(1-\sqrt{2}\right)} + \pi.
\end{equation*}
\end{proof}

\section{Summary}\label{sum}
In this paper we propose an exact synthesis scheme for arbitrary $n$-qubit gates using only SQiSW gates and single-qubit gates. We utilize the algebraic properties of SQiSW gate to optimize the number of SQiSW gates required for synthesizing of an arbitrary 3-qubit gate, achieving a result of 24. Recursively, we provide an upper bound for the synthesis of an n-qubit gate, which is $\frac{139}{192}\times 4^n(1+o(1))$. Also, we state and analyze a pruning algorithm in numerical optimization of the synthesis problems using SQiSW gates, reducing the searching size to $\frac{1}{12}$ of previous size asymptotically. We conduct numerical experiment and find 8 and 11 SQiSW gates are enough for synthesis of Toffoli and arbitrary 3-qubit gates under acceptable error. Finally, we derive an exact synthesis scheme for Toffoli gate using only 8 SQiSW gates from numerical optimization, by observation of circuit parameters.


It remains open whether 8 is a tight result for synthesis of Toffoli gate. That's to say, whether it requires \textit{at least} 8 SQiSW gates to synthesize Toffoli gate. We've already provided numerical evidence in this paper for the problem but the rigorous proof is still absent. Another interesting problem is, whether we can derive exact synthesis schemes in similar numerical-aided way in more complicate synthesis tasks. For example, can we raise an exact synthesis scheme for arbitrary 3-qubit gates using 11, or just less than 24 SQiSW gates by numerical observation? The essential difficulty is that Toffoli gate is a deterministic gate but ``arbitrary'' synthesis task has lots of free parameters in its target circuits. So it's hard to fix parameters explicitly and we must reserve enough free parameters in the final scheme. Also, how the optimization techniques used to achieve the bounds can be generalized is a question worth exploring. For example, one can try to use Weyl chamber to discover more decomposition schemes with special structures (such as the circuit with ``diagonal edge'' in Lemma \ref{lemma6}) that are advantageous for optimization.

\FloatBarrier

\section{Acknowledgment}
The authors thank Longcheng Li, Jiadong Zhu, Ziheng Chen, Qi Ye and Jianxin Chen for helpful discussions. This work was supported in part by the National Natural Science Foundation of China Grants No. 62325210, 92465202, 62272441, and the Strategic Priority Research Program of Chinese Academy of Sciences Grant No. XDB28000000.

\FloatBarrier
\sloppy
\bibliographystyle{quantum}
\bibliography{ref}

\providecommand{\noopsort}[1]{}\providecommand{\singleletter}[1]{#1}%
\begin{thebibliography}{10}

\bibitem{Shende2006}
V.V. Shende, S.S. Bullock, and I.L. Markov.
\newblock ``Synthesis of quantum-logic circuits''.
\newblock \href{https://doi.org/10.1109/TCAD.2005.855930}{IEEE Transactions on Computer-Aided Design of Integrated Circuits and Systems {\bf 25}, 1000–1010}~(2006).

\bibitem{Mottonen2004}
Mikko M\"ott\"onen, Juha~J. Vartiainen, Ville Bergholm, and Martti~M. Salomaa.
\newblock ``Quantum circuits for general multiqubit gates''.
\newblock \href{https://doi.org/10.1103/PhysRevLett.93.130502}{Phys. Rev. Lett. {\bf 93}, 130502}~(2004).

\bibitem{jiang2022}
Jiaqing Jiang, Xiaoming Sun, Shang-Hua Teng, Bujiao Wu, Kewen Wu, and Jialin Zhang.
\newblock ``Optimal space-depth trade-off of {CNOT} circuits in quantum logic synthesis''.
\newblock In Proceedings of the 2020 ACM-SIAM Symposium on Discrete Algorithms (SODA).
\newblock \href{https://doi.org/10.1137/1.9781611975994.13}{Pages 213--229}.
\newblock ~(2020).

\bibitem{Barenco1995}
Adriano Barenco, Charles~H. Bennett, Richard Cleve, David~P. DiVincenzo, Norman Margolus, Peter Shor, Tycho Sleator, John~A. Smolin, and Harald Weinfurter.
\newblock ``Elementary gates for quantum computation''.
\newblock \href{https://doi.org/10.1103/PhysRevA.52.3457}{Phys. Rev. A {\bf 52}, 3457--3467}~(1995).

\bibitem{aho2003}
Alfred~V. Aho and Krysta~M. Svore.
\newblock ``Compiling quantum circuits using the palindrome transform''~(2003).
\newblock  \href{http://arxiv.org/abs/quant-ph/0311008}{arXiv:quant-ph/0311008}.

\bibitem{Amy2018}
Matthew Amy, Parsiad Azimzadeh, and Michele Mosca.
\newblock ``On the controlled-not complexity of controlled-not–phase circuits''.
\newblock \href{https://doi.org/10.1088/2058-9565/aad8ca}{Quantum Science and Technology {\bf 4}, 015002}~(2018).

\bibitem{yang2025}
Shuai Yang, Guojing Tian, Jialin Zhang, and Xiaoming Sun.
\newblock ``Quantum circuit synthesis on noisy intermediate-scale quantum devices''.
\newblock \href{https://doi.org/10.1103/PhysRevA.109.012602}{Phys. Rev. A {\bf 109}, 012602}~(2024).

\bibitem{Cybenko2001}
G.~Cybenko.
\newblock ``Reducing quantum computations to elementary unitary operations''.
\newblock \href{https://doi.org/10.1109/5992.908999}{Computing in Science \& Engineering {\bf 3}, 27--32}~(2001).

\bibitem{vatan2004}
Farrokh Vatan and Colin~P. Williams.
\newblock ``Realization of a general three-qubit quantum gate''~(2004).
\newblock  \href{http://arxiv.org/abs/quant-ph/0401178}{arXiv:quant-ph/0401178}.

\bibitem{Knill1995}
E.~Knill.
\newblock ``Bounds for approximation in total variation distance by quantum circuits''~(1995).
\newblock  \href{http://arxiv.org/abs/quant-ph/9508007}{arXiv:quant-ph/9508007}.

\bibitem{zi2025}
Wei Zi, Junhong Nie, and Xiaoming Sun.
\newblock ``Shallow quantum circuit implementation of symmetric functions with limited ancillary qubits''.
\newblock \href{https://doi.org/10.1109/TCAD.2025.3539002}{IEEE Transactions on Computer-Aided Design of Integrated Circuits and Systems {\bf 44}, 3060--3072}~(2025).

\bibitem{guo}
S.-B. Wang, P.~Wang, G.-H. Li, et~al.
\newblock ``Variational quantum eigensolver with linear depth problem-inspired ansatz for solving portfolio optimization in finance''.
\newblock \href{https://doi.org/10.1007/s11432-024-4185-1}{Sci. China Inf. Sci. {\bf 68}, 180504:1--180504:11}~(2025).

\bibitem{Knill19952}
E.~Knill.
\newblock ``Approximation by quantum circuits''~(1995).
\newblock  \href{http://arxiv.org/abs/quant-ph/9508006}{arXiv:quant-ph/9508006}.

\bibitem{Vartiainen2004}
Juha~J. Vartiainen, Mikko M\"ott\"onen, and Martti~M. Salomaa.
\newblock ``Efficient decomposition of quantum gates''.
\newblock \href{https://doi.org/10.1103/PhysRevLett.92.177902}{Phys. Rev. Lett. {\bf 92}, 177902}~(2004).

\bibitem{shende2004}
V.V. Shende, I.L. Markov, and S.S. Bullock.
\newblock ``Smaller two-qubit circuits for quantum communication and computation''.
\newblock In Proceedings Design, Automation and Test in Europe Conference and Exhibition.
\newblock \href{https://doi.org/10.1109/DATE.2004.1269020}{Volume~2, pages 980--985}.
\newblock ~(2004).

\bibitem{shende2008}
Vivek~V. Shende and Igor~L. Markov.
\newblock ``On the cnot-cost of toffoli gates''~(2008).
\newblock  \href{http://arxiv.org/abs/0803.2316}{arXiv:0803.2316}.

\bibitem{QCQI}
Michael~A. Nielsen and Isaac~L. Chuang.
\newblock ``Quantum computation and quantum information: 10th anniversary edition''.
\newblock \href{https://doi.org/10.1017/CBO9780511976667}{Cambridge University Press}. USA~(2011).
\newblock 10th edition.

\bibitem{nie2024}
Junhong Nie, Wei Zi, and Xiaoming Sun.
\newblock ``Quantum circuit for multi-qubit toffoli gate with optimal resource''~(2024).
\newblock  \href{http://arxiv.org/abs/2402.05053}{arXiv:2402.05053}.

\bibitem{chen2023}
Jianxin Chen, Dawei Ding, Weiyuan Gong, Cupjin Huang, and Qi~Ye.
\newblock ``One gate scheme to rule them all: Introducing a complex yet reduced instruction set for quantum computing''.
\newblock In Proceedings of the 29th ACM International Conference on Architectural Support for Programming Languages and Operating Systems, Volume 2.
\newblock \href{https://doi.org/10.1145/3620665.3640386}{Page 779–796}.
\newblock ASPLOS '24. Association for Computing Machinery~(2024).

\bibitem{Goubault2019}
Timothée Goubault~de Brugière, Marc Baboulin, Benoît Valiron, and Cyril Allouche.
\newblock ``Synthesizing quantum circuits via numerical optimization''.
\newblock \href{https://doi.org/10.1007/978-3-030-22741-8_1}{Page 3–16}.
\newblock Springer International Publishing. ~(2019).

\bibitem{Martinez2016}
Esteban~A Martinez, Thomas Monz, Daniel Nigg, Philipp Schindler, and Rainer Blatt.
\newblock ``Compiling quantum algorithms for architectures with multi-qubit gates''.
\newblock \href{https://doi.org/10.1088/1367-2630/18/6/063029}{New Journal of Physics {\bf 18}, 063029}~(2016).

\bibitem{Cerezo2022}
M.~Cerezo, Kunal Sharma, Andrew Arrasmith, and Patrick~J. Coles.
\newblock ``Variational quantum state eigensolver''.
\newblock \href{https://doi.org/10.1038/s41534-022-00611-6}{npj Quantum Information{\bf 8}}~(2022).

\bibitem{shirakawa2021}
Tomonori Shirakawa, Hiroshi Ueda, and Seiji Yunoki.
\newblock ``Automatic quantum circuit encoding of a given arbitrary quantum state''.
\newblock \href{https://doi.org/10.1103/PhysRevResearch.6.043008}{Phys. Rev. Res. {\bf 6}, 043008}~(2024).

\bibitem{Ashhab2022}
Sahel Ashhab, Naoki Yamamoto, Fumiki Yoshihara, and Kouichi Semba.
\newblock ``Numerical analysis of quantum circuits for state preparation and unitary operator synthesis''.
\newblock \href{https://doi.org/10.1103/PhysRevA.106.022426}{Phys. Rev. A {\bf 106}, 022426}~(2022).

\bibitem{ashhab2023}
Sahel Ashhab, Fumiki Yoshihara, Miwako Tsuji, Mitsuhisa Sato, and Kouichi Semba.
\newblock ``Quantum circuit synthesis via a random combinatorial search''.
\newblock \href{https://doi.org/10.1103/PhysRevA.109.052605}{Phys. Rev. A {\bf 109}, 052605}~(2024).

\bibitem{synthetiq}
Anouk Paradis, Jasper Dekoninck, Benjamin Bichsel, and Martin Vechev.
\newblock ``Synthetiq: Fast and versatile quantum circuit synthesis''.
\newblock \href{https://doi.org/10.1145/3649813}{Proceedings of the ACM on Programming Languages {\bf 8}, 55--82}~(2024).

\bibitem{BQSKit}
Siyuan Niu, Akel Hashim, Costin Iancu, Wibe~Albert De~Jong, and Ed~Younis.
\newblock ``Effective quantum resource optimization via circuit resizing in bqskit''.
\newblock In Proceedings of the 61st ACM/IEEE Design Automation Conference.
\newblock \href{https://doi.org/10.1145/3649329.3656534}{DAC '24}. Association for Computing Machinery~(2024).

\bibitem{BQSKit_code}
Ed~Younis and Emma Smith.
\newblock ``Bqskit github home page''.
\newblock \url{https://github.com/BQSKit}.

\bibitem{Norbert2003}
Norbert Schuch and Jens Siewert.
\newblock ``Natural two-qubit gate for quantum computation using the $\mathrm{XY}$ interaction''.
\newblock \href{https://doi.org/10.1103/PhysRevA.67.032301}{Phys. Rev. A {\bf 67}, 032301}~(2003).

\bibitem{Bialczak2010}
R.~C. Bialczak, M.~Ansmann, M.~Hofheinz, E.~Lucero, M.~Neeley, A.~D. O’Connell, D.~Sank, H.~Wang, J.~Wenner, M.~Steffen, A.~N. Cleland, and J.~M. Martinis.
\newblock ``Quantum process tomography of a universal entangling gate implemented with josephson phase qubits''.
\newblock \href{https://doi.org/10.1038/nphys1639}{Nature Physics {\bf 6}, 409–413}~(2010).

\bibitem{Abrams2020}
Deanna~M. Abrams, Nicolas Didier, Blake~R. Johnson, Marcus P.~da Silva, and Colm~A. Ryan.
\newblock ``Implementation of xy entangling gates with a single calibrated pulse''.
\newblock \href{https://doi.org/10.1038/s41928-020-00498-1}{Nature Electronics {\bf 3}, 744–750}~(2020).

\bibitem{Huang2023}
Cupjin Huang, Tenghui Wang, Feng Wu, Dawei Ding, Qi~Ye, Linghang Kong, Fang Zhang, Xiaotong Ni, Zhijun Song, Yaoyun Shi, Hui-Hai Zhao, Chunqing Deng, and Jianxin Chen.
\newblock ``Quantum instruction set design for performance''.
\newblock \href{https://doi.org/10.1103/PhysRevLett.130.070601}{Phys. Rev. Lett. {\bf 130}, 070601}~(2023).

\bibitem{bullock2003}
Stephen~S. Bullock and Igor~L. Markov.
\newblock ``Smaller circuits for arbitrary n-qubit diagonal computations''~(2003).
\newblock  \href{http://arxiv.org/abs/quant-ph/0303039}{arXiv:quant-ph/0303039}.

\bibitem{tucci2005}
Robert~R. Tucci.
\newblock ``An introduction to cartan's kak decomposition for qc programmers''~(2005).
\newblock  \href{http://arxiv.org/abs/quant-ph/0507171}{arXiv:quant-ph/0507171}.

\bibitem{byron}
Byron Drury and Peter Love.
\newblock ``Constructive quantum shannon decomposition from cartan involutions''.
\newblock \href{https://doi.org/10.1088/1751-8113/41/39/395305}{Journal of Physics A: Mathematical and Theoretical {\bf 41}, 395305}~(2008).

\bibitem{cartan}
Navin Khaneja and Steffen~J. Glaser.
\newblock ``Cartan decomposition of su(2n) and control of spin systems''.
\newblock \href{https://doi.org/10.1016/S0301-0104(01)00318-4}{Chemical Physics {\bf 267}, 11--23}~(2001).

\bibitem{wyle}
Andrew~W. Cross, Lev~S. Bishop, Sarah Sheldon, Paul~D. Nation, and Jay~M. Gambetta.
\newblock ``Validating quantum computers using randomized model circuits''.
\newblock \href{https://doi.org/10.1103/PhysRevA.100.032328}{Phys. Rev. A {\bf 100}, 032328}~(2019).

\bibitem{geo}
Jun Zhang, Jiri Vala, Shankar Sastry, and K.~Birgitta Whaley.
\newblock ``Geometric theory of nonlocal two-qubit operations''.
\newblock \href{https://doi.org/10.1103/PhysRevA.67.042313}{Phys. Rev. A {\bf 67}, 042313}~(2003).

\bibitem{tucci1998}
Robert~R. Tucci.
\newblock ``A rudimentary quantum compiler(2cnd ed.)''~(1999).
\newblock  \href{http://arxiv.org/abs/quant-ph/9902062}{arXiv:quant-ph/9902062}.

\bibitem{qfactor}
Alon Kukliansky, Ed~Younis, Lukasz Cincio, and Costin Iancu.
\newblock ``Qfactor: A domain-specific optimizer for quantum circuit instantiation''.
\newblock In 2023 IEEE International Conference on Quantum Computing and Engineering (QCE).
\newblock \href{https://doi.org/10.1109/qce57702.2023.00096}{Page 814–824}.
\newblock IEEE~(2023).

\bibitem{Haar}
Francesco Mezzadri.
\newblock ``How to generate random matrices from the classical compact groups''~(2007).
\newblock  \href{http://arxiv.org/abs/math-ph/0609050}{arXiv:math-ph/0609050}.

\bibitem{Yu_2013}
Nengkun Yu, Runyao Duan, and Mingsheng Ying.
\newblock ``Five two-qubit gates are necessary for implementing the toffoli gate''.
\newblock \href{https://doi.org/10.1103/PhysRevA.88.010304}{Phys. Rev. A {\bf 88}, 010304}~(2013).

\bibitem{Shen2004}
Vivek~V. Shende, Igor~L. Markov, and Stephen~S. Bullock.
\newblock ``Minimal universal two-qubit controlled-not-based circuits''.
\newblock \href{https://doi.org/10.1103/PhysRevA.69.062321}{Phys. Rev. A {\bf 69}, 062321}~(2004).

\end{thebibliography}

\end{document}